\documentclass[conference,a4paper]{IEEEtran}

\usepackage[utf8]{inputenc} 
\usepackage[T1]{fontenc}
\usepackage{url}
\usepackage{ifthen}
\usepackage{cite}
\usepackage[cmex10]{amsmath} 

\usepackage{amsmath,amssymb,algorithm,cite,caption,cases,bm,float,graphicx,url,color}
\usepackage{thm-restate, enumerate }
\usepackage{amsthm}

\newboolean{showcomments}
\setboolean{showcomments}{true}
\newcommand{\ehsan}[1]{  \ifthenelse{\boolean{showcomments}}
{ \textcolor{blue}{(Ehsan says:  #1)}} {}  }
\newcommand{\fari}[1]{\ifthenelse{\boolean{showcomments}}
{ \textcolor{magenta}{(Fariborz says: #1)} } {} }
\newcommand{\babak}[1]{\ifthenelse{\boolean{showcomments}}
{ \textcolor{green}{(Babak says:  #1)}}{}}

\newcommand{\Sc}{{\mathbf{S}}}

\newcommand{\al}{{\alpha}}

\newcommand{\re}{{\text{Re}}}
\newcommand{\im}{{\text{Im}}}
\newcommand{\xin}{\mathbf{ x}_{\text{init}}}
\newcommand{\derec}{\delta_{\text{rec}}}
\newcommand{\psiubw}{\psi(u,w)}
\newcommand{\rhoin}{\rho_{\text{init}}}

\newcommand{\dist}{\mathbf{D}}

\newcommand{\xn}{\mathbf{x}_0}
\newcommand{\xh}{\hat{\mathbf{x}}}

\newcommand{\sig}{\sigma}
\newcommand{\R}{\mathbb{R}}
\newcommand{\C}{\mathbb{C}}

\newcommand{\x}{\mathbf{x}}
\newcommand{\w}{\mathbf{w}}

\newcommand{\h}{\mathbf{h}}
\newcommand{\g}{\mathbf{g}}
\newcommand{\av}{\mathbf{a}}

\newcommand{\ub}{\mathbf{u}}

\newcommand{\dv}{\mathbf{d}}

\newcommand{\G}{\mathbf{G}}

\newcommand{\Nc}{\mathcal{N}}
\newcommand{\Ncc}{{\mathcal{N}_\mathbb{C}}}

\usepackage{dsfont}

\theoremstyle{theorem}

\newtheorem{thm}{Theorem}[section]

\newtheorem{lem}{Lemma}[section]

\newtheorem{cor}{Corollary}[section]

\theoremstyle{remark}
\newtheorem{remark}{Remark}

\theoremstyle{definition}

\newcommand{\beq}{\begin{equation}}
\newcommand{\eeq}{\end{equation}}
\newcommand{\bea}{\begin{align}}
\newcommand{\eea}{\end{align}}

\begin{document}

\title{A Precise Analysis of PhaseMax in Phase Retrieval} 


\author{%
  \IEEEauthorblockN{Fariborz Salehi, Ehsan Abbasi, Babak Hassibi}
  \IEEEauthorblockA{Department of Electrical Engineering, Caltech, Pasadena, CA 91125\\
                    Email: \{fsalehi, eabbasi, hassibi\}@caltech.edu}

}


\maketitle

\begin{abstract}
   Recovering an unknown complex signal from the magnitude of linear combinations of the signal is referred to as phase retrieval. We present an exact performance analysis of a recently proposed convex-optimization-formulation for this problem, known as PhaseMax. Standard convex-relaxation-based methods in phase retrieval resort to the idea of ''lifting'' which makes them computationally inefficient, since the number of unknowns is effectively squared. In contrast, PhaseMax is a novel convex relaxation that does not increase the number of unknowns. Instead it relies on an initial estimate of the true signal which must be externally provided. In this paper, we investigate the required number of measurements for exact recovery of the signal in the large system limit and when the linear measurement matrix is random with iid standard normal entries. If $n$ denotes the dimension of the unknown complex signal and $m$ the number of phaseless measurements, then in the large system limit, $\frac{m}{n} > \frac{4}{\cos^2(\theta)}$ measurements is necessary and sufficient to recover the signal with high probability, where $\theta$ is the angle between the initial estimate and the true signal. Our result indicates a sharp phase transition in the asymptotic regime which matches the empirical result in numerical simulations.
\end{abstract}

\section{Introduction}

The fundamental problem of recovering a signal from magnitude-only measurements is known as \textit{phase retrieval}. This problem has a rich history and occurs in many areas in engineering and applied physics such as  astronomical imaging \cite{Astro2}, X-ray crystallography \cite{Crystal}, medical imaging \cite{Medical2}, and optics \cite{Optic}. In most of these cases, measuring the phase is either expensive or even infeasible. For instance, in some optical settings, detection devices like CCD cameras and photosensitive films cannot measure the phase of a light wave and instead measure the photon flux. 

Reconstructing a signal from magnitude-only measurements is generally very difficult due to loss of important phase information. Therefore, phase retrieval faces fundamental theoretical
and algorithmic challenges and a variety of methods were suggested~\cite{jaganathan2015phase}. Convex methods have recently gained significant attention to solve the phase retrieval problem. These methods are mainly based on semidefinite programming by linearizing the resulting quadratic constraints using the idea of \textit{lifting}~\cite{candes2015phase, K6, LIFTING1,LIFTING2, walk2017blind,CANDESPL, BALAN, Fariborz17, SAMET, ROMBERG, jaganathan2012recovery}. Due to the convex nature of their formulation, these algorithms usually have rigorous theoretical guarantees. However, semidefinite relaxation squares the number of unknowns which makes these algorithms computationally complex, especially in large systems. This caveat makes these approaches intractable in real-world applications.

Introduced in two independent works~\cite{goldstein2016phasemax, bahmani2016phase}, \textit{PhaseMax} is a recently proposed convex formulation for the phase retrieval problem in the original $n-$dimensional parameter space. This method maximizes a linear functional over a  convex feasible set. The constrained set in this optimization is obtained by relaxing the non-convex equality constraints in the original phase retrieval problem to convex inequality constraints. To form the objective function, PhaseMax relies on an initial estimate of the true signal which must be externally provided.

The simple formulation of the PhaseMax method makes it appealing for practical applications. In addition, existing theoretical analysis  indicates  this method achieves perfect recovery for a nearly optimal number of random measurements.  The analysis in~\cite{goldstein2016phasemax, bahmani2016phase, hand2016elementary} suggests that $m>Cn$, where $C$ is a constant that depends on the quality of initial estimate ($\xin$), is the sufficient number of measurements for perfect signal reconstruction when the measurement vectors are drawn independently from the Gaussian distribution. The exact phase transition threshold, i.e. the exact value of the constant $C$, for the \emph{real} PhaseMax has been recently derived in ~\cite{dhifallah2017phase, dhifallah2017fundamental}. However, for the practical case of complex signals, previous results could only provide an upper bound on $C$. 

In this paper, we characterize the phase transition regimes for the perfect signal recovery in the PhaseMax algorithm. Our result is asymptotic and assumes that the measurement vectors are derived independently from Gaussian distribution. To the extent of our knowledge, this is the first work that computes the exact phase transition bound of the (complex-valued) PhaseMax in phase retrieval. 

In our analysis, we utilize the recently developed Convex Gaussian Min-max Theorem (CGMT)~\cite{cgmt} which uses Gaussian process methods. CGMT has been successfully applied in a number of different problems including the performance analysis of structured signal recovery in M-estimators \cite{cgmt,abbasi2016general}, massive MIMO \cite{berreal,massivemimo} and etc. CGMT has been also used by Dhifallah et. al. \cite{dhifallah2017phase} to analyze the real version of the PhaseMax. But unfortunately, the complex case does not directly fit into the framework of CGMT. Therefore, in this paper we introduce a secondary optimization that provably has the same phase transition bounds as PhaseMax and that also can be analyzed by CGMT. 

The organization of the paper is as follows. In section \ref{sec:setup}  we introduce the main notations and mathematically setup the problem. In section \ref{sec:thm}, we present our main result followed by discussions and the result of numerical simulations. Finally, section \ref{sec:proof} includes an outline of the proof of the main theorem. 
\section{Problem Setup}
\label{sec:setup}

\subsection{Notations}

We gather here the basic notations that are used throughout this paper. We reserve the letter $j$ for the complex unit. For a complex scalar $x\in\C$,  $x_\re$ and $x_\im$ correspond to the real and imaginary parts of $x$, respectively, and  $|x|=\sqrt{x_\re^2+x_\im^2}\;$. $\Nc(\mu,\sig^2)$ denotes real Gaussian distribution with mean $\mu$ and variance $\sig^2$. Similarly, $\Ncc(\mu,\sig^2)$ refers to a \textit{complex} Gaussian distribution with real and imaginary parts drawn independently from $\Nc(\mu_\re,\sig^2/2)$ and $\Nc(\mu_\im,\sig^2/2)$, respectively. $\mathcal R(2\sigma^2)$ denotes the Rayleigh distribution with second moment equal to $2\sigma^2$. $X\sim p_X$ implies that the random variable $X$ has a density $p_{X}$. Bold lower letters are reserved for vectors and upper letters are used for matrices. For a vector $\mathbf v$, $v_i$ denotes its $i^{\text{th}}$ entry and $||\mathbf v||$ is its $l_2$ norm. $(\cdot)^{\star}$ is used to denote the conjugate transpose. For a complex vector $\mathbf v$, ${\mathbf v}_{\re}$ and ${\mathbf v}_{\im}$ denotes its real and complex parts, respectively. Also, ${\mathbf v}(k:l)$ is a column vector consisting of entries with index from $k$ to $l$ of $\mathbf v$. We use caligraphy letters for sets. For set $\mathcal S$, $\text{cone}(\mathcal S)$ is the closed conical hull of $\mathcal S$.  

\subsection{Setup}

Let $\x_0\in \C^n$ denote the underlying signal. We consider the phase retrieval problem with the goal of recovering $\x_0$ from $m$ magnitude-only measurements of the form,
\begin{align}\label{eq::measurements}
b_i=|\mathbf \av_i^\star \x_0|,~ i=1,\dots,m.
\end{align} 
Throughout this paper we assume that  $\{\mathbf \av_i \in \C^n\}_{i=1}^m$  is the set of known measurement vectors where the $\av_i$'s are independently drawn from the complex Gaussian distribution with mean zero and covariance matrix $\mathbf I$. 

As mentioned earlier, the PhaseMax method relies on an initial estimate of the true signal.  $\xin\in\C^n$  is used to represent this initial guess. We assume both $\xn$ and $\xin $ are independent of all the measurement vectors. The PhaseMax algorithm provides a convex formulation of the phase retrieval problem by simply relaxing the equality constraints in \eqref{eq::measurements} into \textit{convex} inequality constraints. This results in the following convex optimization problem:

\begin{equation}
\label{eq::phasemax}
\begin{aligned}
&\xh=\arg\max_{\x\in\C^n}&&\re\{\mathbf{\xin}^{\star}\;\mathbf{x}\}\\
&{\text{subject to:}}&&|\mathbf a_i^{\star}\mathbf x|\leq b_i\;,\;\;1\leq i \leq m.
\end{aligned}
\end{equation}
This optimization searches for a feasible vector that posses the most real correlation with $\xin$. Note that because of the global phase ambiguity of the measurements in \eqref{eq::measurements}, we can estimate $\xn$ up to a global phase. Therefore, we define the following performance measure for the PhaseMax method,
\begin{align}\label{eq::performance}
&\dist(\xh,\xn)=\min_{\phi\in [-\pi,\pi]}\frac{\| \xh e^{j\phi}-\xn  \|}{\|\xn\|}\;.
\end{align}
Under this setting, a perfect recovery of $\xn$ means $\dist(\xh,\xn)=0$. In this paper we investigate the necessary and sufficient conditions under which the optimization program \eqref{eq::phasemax} perfectly recovers the true signal.

\section{Main Result}
\label{sec:thm}

In this section, we present the main result of the paper which provides us with the necessary and sufficient number of measurements for the perfect recovery of the PhaseMax method in \eqref{eq::phasemax} under different scenarios. Our result is asymptotic which assumes a fixed oversampling ratio $\delta:=\frac{m}{n}\in[0,\infty)$, while $n\rightarrow \infty$. In theorem \ref{thm::phasetransition}, we introduce $\derec$ which depends on the problem parameters and prove that the condition $\delta>\derec$, is necessary and sufficient for perfect recovery. Our result reveals significant dependence between $\derec$ and the quality of the initial guess. We use the following similarity measure to quantify the caliber of the initial estimate:
\begin{equation}\label{eq::init_measure}
\rho_{\text{init}}:=\max_{0\leq \phi<2\pi} \frac{\re\{e^{j\phi}\; \xin^{\star} \;\x_0\}} {||\x_0|| \;||\xin ||}\;=\;\frac{|\xin^{\star} \;\x_0|}{||\x_0|| \;||\xin ||}.
\end{equation} 
Note that the multiplication by a unit amplitude scalar in the above definition is due to the global phase ambiguity of the phase retrieval solution (the true phase of $\xn$ is dissolved in the absolute value in \eqref{eq::measurements}). Therefore, for convenience we assume both $\xin$ and $\x_0$ are aligned unit norm vectors ($||\x_0||=||\xin||=1$), which results in $\rho_{\text{init}}=\xin^{\star} \;\x_0$. We also define $\theta$ as the angle between $\xin$ and $\x_0$, and therefore, $\rho_{\text{init}}=\cos \theta$. We now present the main result of the paper which characterizes the phase transition regimes of PhaseMax for perfect recovery, in terms of $\delta$ and $\rhoin$.

\begin{thm} 
\label{thm::phasetransition}
Consider the PhaseMax problem defined in section \ref{sec:setup}. For a fixed oversampling ratio $\delta=\frac{m}{n}>4$, the optimization program \eqref{eq::phasemax} perfectly recovers the true signal (in the sense that $\lim_{n\rightarrow \infty} \mathbb P(\dist(\xh,\xn)>\epsilon)=0$, for any fixed $\epsilon>0$) if and only if,

\begin{equation}
\label{eq::bound}
\delta > \derec:=\frac{4}{\cos^2\theta}=\frac{4}{\rho_{\text{init}}^2}\;,
\end{equation}
where $\rho_{\text{init}}$ is defined in \eqref{eq::init_measure}.
\end{thm}

Theorem \ref{thm::phasetransition} establishes a sharp phase transition behavior for the performance of PhaseMax. The inequality \eqref{eq::bound} can also be rewritten in terms of $\theta$ (or $\rho_{\text{init}}$) when the oversampling ratio, $\delta$, is fixed, 

\begin{equation}
\label{eq::boundrho}
\rho_{\text{init}}=\cos\theta>\sqrt{\frac{4}{\delta}}\; .
\end{equation}


The proof of Theorem \ref{thm::phasetransition} consists of two main steps. First, we introduce a real optimization program with $2n-1$ variables and prove that it has the same phase transition bounds as PhaseMax in \eqref{eq::phasemax}. The point of this step is that this new real optimization is especially built in a way that its performance can be precisely analyzed using well known tools like CGMT. Therefore, the next step would be to apply the CGMT framework to the new real optimization and to derive its phase transition bounds.  We postpone a detailed version of the proof to section \ref{sec:proof}.
\begin{remark}\label{rem::rem1}
The condition $\delta>4$ is proven to be fundamentally necessary for the phase retrieval problem under generic measurements to have a unique solution \cite{conca2015algebraic}. This is consistent with Theorem \ref{thm::phasetransition} where you can observe that even in the best scenario where $\xin$ is aligned with $\xn$, we still need $m>4n$ measurements for PhaseMax to have $\xn$ as the solution. On the other hand, in the case where $\xin$ carries no information about $\xn$ ($\xin$ is orthogonal to $\xn$), recovery of $\xn$ by PhaseMax is not guaranteed regardless of the number of measurements. 
\end{remark}

\begin{remark}
It is shown in the work of Goldstein et. al. \cite{goldstein2016phasemax} that $\delta>\frac{4}{1-2\theta/\pi}$ is \textit{sufficient} for perfect recovery of $\xn$. This bound is compared to our result in Fig. \ref{fig1} which shows phase transition regions of PhaseMax derived from empirical results. Although the simulations are run on the signals of size $n=128$, one can see that the blue line that comes from Theorem \ref{thm::phasetransition}, perfectly predicts phase transition boundary.
\end{remark}

\begin{figure}[htp]
\label{fig1}
\begin{center}
\includegraphics[width=8cm]{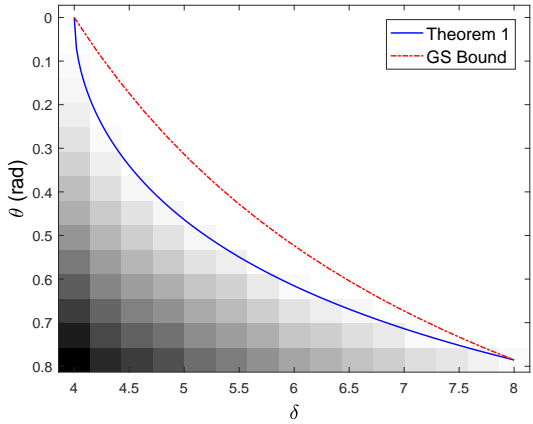}
\end{center}
  \caption{Phase transition regimes for the PhaseMax problem in terms of the oversampling ratio $\delta=m/n$ and $\theta$, the angle between $\xn$ and $\xin$. For the empirical results, we used signals of size $n=128$. The data is averaged over 10 independent realization of the measurement vectors. The blue line indicates the sharp phase transition bounds derived in Theorem \ref{thm::phasetransition} and the red line comes from the results of  \cite{goldstein2016phasemax}, which is referred to as the GS Bound.}
  \label{fig1}
\end{figure}

\section{Proof Outline}\label{sec:proof}

In this part we introduce the main ideas used in the proof of Theorem \ref{thm::phasetransition}. As mentioned earlier in section \ref{sec:thm}, we assume $\x_0$  is a unit norm vector aligned with $\xin$. Due to rotational invariance of the Gaussian distribution, without loss of generality, we assume $\x_0=\mathbf{e_1}$, the first vector of the standard basis in $\mathbb \C^n$. Furthermore, the optimization program \eqref{eq::phasemax} is scalar invariant. So, we can assume $\|\xin\|=1$.

The proof consists of two main steps: In the first step, we analyze the complex optimization problem \eqref{eq::phasemax} and find the necessary and sufficient condition under which $\hat{\mathbf x}=\x_0$. Consequently, we use this condition to build an equivalent real optimization problem. Lemma \ref{lemma_real} introduces this equivalent real optimization ERO, in $\mathbb R^{2n-1}$, and states that the perfect recovery in the PhaseMax algorithm occurs if and only if zero is the unique minimizer of the ERO. 

In the second step, we adopt the CGMT framework to analyze the ERO and investigate the conditions on $\rho_{\text{init}}$ (or $\theta$) under which the unique answer to the ERO is $\mathbf 0$. Therefore ,as a result of Lemma \ref{lemma_real}, these conditions will guarantee the perfect recovery in the initial PhaseMax optimization \eqref{eq::phasemax}.

\subsection{Introducing the Real Optimization ERO }\label{sub:ero}
We define the error vector $\mathbf w:=\mathbf x-\x_0$ and rewrite \eqref{eq::phasemax} in terms of $\mathbf w$,

\begin{equation}
\label{eq::phasemax2}
\begin{aligned}
&\max_{\w\in\C^n}&&\re\{\mathbf{\xin}^{\star}\;\mathbf{w}\}\\
&{\text{subject to:}}&&|\mathbf a_i^{\star}\mathbf (\mathbf e_1+\mathbf w)|\leq b_i\;,\;\;1\leq i \leq m.
\end{aligned}
\end{equation}

For $i=1,2,\ldots,m$ , we use $\phi_i:=\text{phase}(\mathbf{a_i}^{\star}\x_0)$ to define aligned measurement vectors $\tilde{\mathbf a}_i:=e^{j\phi_i}\mathbf a_i$. Therefore, we have,

\begin{equation}
b_i={\tilde{\mathbf{a}}}_i^{\star}\x_0=({\tilde{\mathbf{a}}}_i)_1,\;\;\; \text{for}\;\; i=1,2,\ldots,m\;,
\end{equation}
where $({\tilde{\mathbf{a}}}_i)_1$ is the first entry of ${\tilde{\mathbf{a}}}_i$.
Let $\mathcal D:=\{\mathbf w\in \C^n: \re\{\xin^{\star}\;\mathbf{w}\}\geq 0\}$ be the set of all vectors $\mathbf w$ with nonnegative objective value and $\mathcal F:=\{\mathbf w\in \C^n: |\mathbf a_i^{\star}\mathbf (\mathbf e_1+\mathbf w)|\leq b_i,\;\text{for}\;i=1,2,\dots,m\}$ be the feasible set of the optimization problem \eqref{eq::phasemax2}. The following lemmas prove necessary and sufficient conditions for perfect recovery in PhaseMax, based on these notations.

\begin{lem}
\label{lemma1}
$\x_0$ is the unique optimal solution of \eqref{eq::phasemax} if and only if $\mathcal D\bigcap \mathcal F=\{\mathbf 0\}$.
\end{lem}

\begin{proof}
For $\mathbf w \in \mathcal D\bigcap \mathcal F$, $\x_0+\mathbf w$ is a solution of \eqref{eq::phasemax} with an  objective value greater than the value for $\x_0$. Therefore, $\mathcal D\bigcap \mathcal F=\{\mathbf 0\}$ is equivalent to $\x_0$ be a local minimum of \eqref{eq::phasemax} which is also a global minimum due to convexity of \eqref{eq::phasemax}.
\end{proof}

\begin{lem}
\label{lemma2}
$\mathcal D\bigcap \mathcal F=\{\mathbf 0\}$ if and only if $\mathcal D\bigcap \text{cone}(\mathcal F)=\{\mathbf 0\}$.
\end{lem}

\begin{proof}
Note that $\mathcal D\subset \C^n$ is a convex cone and $\mathcal F\subset \C^n$ is a convex set. The proof is the consequence of the following equality,
$$\mathcal D\bigcap \text{cone}(\mathcal F)=\text{cone}(\mathcal D\bigcap \mathcal F).$$

\end{proof}

\begin{lem}
\label{lemma3}

$\text{cone}(\mathcal F)=\bigcap_{i=1}^{m}\{\mathbf w\in \C^n: \re\{{{\tilde{\mathbf{a}}}_i}^{\star}\;\mathbf w\}\leq0\}.$

\end{lem}

\begin{proof}
Let $\dv\in \mathcal F$, 

\begin{equation}
\label{eq::constraints}
|b_i+{\tilde{\mathbf{a}}}_i^{\star}\dv|\leq b_i \;,\;\;\text{for}\;\; i=1,2,\ldots,m.
\end{equation}
Therefore, 
\begin{align}
\re\{{\tilde{\mathbf{a}}}_i^{\star}\dv\}\;&= \re\{b_i+{\tilde{\mathbf{a}}}_i^{\star}\dv\}-b_i\;,\nonumber\\
 &\leq |b_i+{\tilde{\mathbf{a}}}_i^{\star}\dv|-b_i\;,\\
 &\leq 0\;\nonumber.
\end{align}
This shows that $\text{cone}(\mathcal F)\subseteq \bigcap_{i=1}^{m}\{\mathbf w\in \C^n: \re\{{\tilde{\mathbf{a}}}_i^{\star}\mathbf w\}\leq 0\}$.
To show the other direction, choose $\dv\in \C^n$ such that: $\re\{{\tilde{\mathbf{a}}}_i^{\star}\dv\}<0,\;\text{for}\;\;i=1,2,\ldots,m.$ One can show that there exists $R>0$, such that for all $r\leq R$, $r\dv\in \mathcal F$. Therefore, $\dv\in \text{cone}(\mathcal F)$. This concludes the proof.

\end{proof}

We have the following corollary as a result of Lemma~\ref{lemma1}, Lemma~\ref{lemma2}, and  Lemma~\ref{lemma3}. 

\begin{cor}
$\x_0$ is the unique optimal solution of \eqref{eq::phasemax} if and only if,
\begin{equation}
\label{optcond.}
\{\mathbf w: \re\{\xin^{\star}\mathbf w\}\geq 0\;\; \re\{{\tilde{\mathbf{a}}}_i^{\star}\mathbf w\}\leq 0, \;\text{ for}\;\;1\leq i \leq m\}=\{\mathbf 0\}.
\end{equation}
\end{cor}
We are now ready to establish the equivalent real optimization ERO. We will show that the ERO has the exact phase transition bounds as PhaseMax in \eqref{eq::phasemax}.
\begin{equation}
\label{eq::ERO}
\begin{aligned}
&\max_{\w'\in\R^{2n-1}}&&\mathbf\eta^{T}\;\mathbf{\w'}\\
&{\text{subject to:}}&&|\mathbf {a'}_i^{T}\mathbf (\mathbf e_1+\mathbf w')| \leq b_i\;,\;\;1 \leq i \leq m,
\end{aligned}
\end{equation}
where $\mathbf e_1$ is the first vector of the standard basis in $\mathbb R^{2n-1}$, $\mathbf{\eta}$ and $\{\mathbf a'_i\}_{i=1}^m$ are $(2n-1)$ dimensional real vectors defined as,
\begin{equation}
\label{eq:real equivalent}
\mathbf{\eta}:=\begin{bmatrix} 
\re\{\xin\}\\
-\im\{\xin(2:n)\}
\end{bmatrix}\;\text{and}\;\mathbf a'_i:=\begin{bmatrix} 
\re\{{\tilde{\mathbf{a}}}_i\}\\
-\im\{{\tilde{\mathbf{a}}}_i(2:n)\}
\end{bmatrix},\;\forall i.
\end{equation} 
Here $\im\{{\tilde{\mathbf{a}}}_i(2:n)\}$ is the imaginary part of the last $n-1$ entries of $\tilde{\mathbf{a}}_i$. We conclude this step of the proof with the following lemma:

\begin{lem}
\label{lemma_real}
$\x_0$ is the unique optimal solution of the PhaseMax method if and only if $\w'=0$ is the unique optimal solution of \eqref{eq::ERO}.
\end{lem}

The proof of Lemma \ref{lemma_real} is straightforward by defining

\begin{equation}
\w'=
\begin{bmatrix}
\re\{\w\}\\
\im\{\w(2:n)\}
\end{bmatrix}\in \mathbb R^{2n-1}\;,
\end{equation}
and then showing that the optimality conditions for $\w'=0$ in \eqref{eq::ERO} is equivalent to \eqref{optcond.}.
 
 It is worth mentioning that the result of Lemma \ref{lemma_real} is valid for any set of measurement vectors $\{\mathbf{a}_i\}$. In the next part, we use this result to compute the phase transition of PhaseMax when the measurement vectors are drawn independently from the Gaussian distribution.
 
 \subsection{Convex Gaussian Min-Max Theorem}
Our analysis is based on the recently developed Convex Gaussian Min-max Theorem (CGMT) \cite{cgmt}. The CGMT associates with a Primary Optimization (PO) problem an Auxiliary Optimization (AO) problem from which we can investigate various properties of the primary optimization, such as phase transitions. In particular, the (PO) and the (AO) problems are defined respectively as follows:
\begin{subequations}\label{eq:POAO}
\begin{align}
\label{eq:PO_gen}
\Phi(\G)&:= \min_{\w\in\Sc_\w}~\max_{\ub\in\Sc_\ub}~ \ub^T\G\w + \psiubw,\\
\label{eq:AO_gen}
\phi(\g,\h)&:= \min_{\w\in\Sc_\w}~\max_{\ub\in\Sc_\ub}~ \|\w\|\g^T\ub - \|\ub\|\h^T\w + \psiubw,
\end{align}
\end{subequations}
where $\G\in\R^{m\times n}, \g\in\R^m, \h\in\R^n$, $\Sc_\w\subset\R^n,\Sc_\ub\subset\R^m$ and $\psi:\R^n\times\R^m\rightarrow\R$. Denote $\w_\Phi:=\w_\Phi(\G)$ and $\w_\phi:=\w_\phi(\g,\h)$ any optimal minimizers in \eqref{eq:PO_gen} and \eqref{eq:AO_gen}, respectively.  The following lemma is a result of CGMT \cite{cgmt}.
\begin{lem}\label{lem:CGMT}
Consider the two optimizations \eqref{eq:PO_gen} and \eqref{eq:AO_gen}. Let $\Sc_\w,\Sc_\ub$ be convex and compact sets, $\psi$ be continuous and convex-concave on $\Sc_\w\times\Sc_\ub$, and, $\G,\g$ and $\h$ all have entries iid standard normal.
Suppose there exist $\al$ such that in the limit of $n\rightarrow\infty$ it holds in probability that $\|\w_\phi(\g,\h)\| \rightarrow\al$. Then, the same holds for $\w_\Phi(\G)$ and we have $\|\w_\Phi(\G)\| \rightarrow\al$.
\end{lem}
In the next section, first we will rewrite the ERO in the form of the optimization \eqref{eq:PO_gen}. This enables us to apply Lemma \ref{lem:CGMT} to the ERO and derive an Auxiliary Optimization in the form of \eqref{eq:AO_gen}. This lemma indicates that if $\|\w_\phi(\g,\h)\| \rightarrow0$ for the (AO), then $\|\w_\Phi(\G)\| \rightarrow0$ for the ERO and we have perfect recovery. (AO) can be analyzed using the conventional concentration results in high dimensions.

 \subsection{Computing the Phase Transition for PhaseMax}
 
 In this part we adopt the CGMT framework along with the result of Lemma \ref{lemma_real} to compute the exact phase transition of the PhaseMax algorithm under the Gaussian measurement scheme.
 
We start by calculating the distribution of the entries of $\mathbf a'_i$ that are defined in \eqref{eq:real equivalent}. Recall that $\mathbf a_i$'s are independently drawn from the complex Gaussian distribution with mean zero and covariance matrix $\mathbf I$. Therefore, the distribution of the entries of ${\tilde{\mathbf{a}}}_i$'s that were defined in section \ref{sub:ero}, is as follows:

 \begin{enumerate}
 \item The first entry of ${\tilde{\mathbf{a}}}_i$ is the absolute value of the first entry of the ${{\mathbf{a}}}_i$. Therefore, it has a Rayleigh distribution, i.e.,
 \begin{equation}
 ({\tilde{\mathbf{a}}}_i)_1\sim \mathcal R(1),
 \end{equation}
 \item The remaining entries of ${\tilde{\mathbf{a}}}_i$ remain standard Gaussian random variables,
  \begin{equation}
 ({\tilde{\mathbf{a}}}_i)_k\sim \Ncc(0,1),\;\; \text{for}\;2\leq k\leq n\;,
 \end{equation}
 \item The entries of ${\tilde{\mathbf{a}}}_i$ remain independent.
 \end{enumerate}
This implies that all the entries of $\mathbf a'_i$ are independent, the first entry of $\mathbf a'_i$ has a $\mathcal R(1)$ distribution and the rest of the entries have Gaussian distribution $\Nc(0,\frac{1}{2})$.
We form the measurement matrix $\mathbf{A}\in \mathbb R^{m\times (2n-1)}$ by stacking vectors $\{\mathbf{a_i}^T,\;1\leq i\leq m\}$. Let $\mathbf{A}_1\in \R^m$ be the first column of $\mathbf A$, and $\tilde{\mathbf A}\in \mathbb R^{m\times(2n-1)}$ be the remaining part (i.e., $\mathbf A=[{\mathbf A}_1\;\;\tilde{\mathbf A}]$). $\x_0=\mathbf e_1$ implies that $\mathbf A_1=[b_1,b_2,\ldots, b_m]^T$, where $b_i$'s are defined in \eqref{eq::measurements}. Using the Lagrange multipliers, we can reformulate \eqref{eq::ERO} as the following minmax program,

\begin{align}
\label{eq::ERO2}
\min_{\substack{w_1\in\R\\ \tilde{\w}\in\R^{2n-2}}}\max_{\mathbf{\lambda},\mathbf{\mu} \in \R_{+}^m}&
\big( -\mathbf{\eta}^{T}\w+{(\mathbf {\lambda}-\mathbf {\mu})}^{T}\tilde{\mathbf{A}}\tilde{\w}\nonumber\\
&-{(\mathbf {\lambda}+\mathbf {\mu})}^{T}\mathbf{A}_1+{(\mathbf {\lambda}-\mathbf {\mu})}^{T}\mathbf{A}_1(1+w_1)\big),
\end{align}
where $w_1$ denotes the first entry of $\w$ and $\tilde{\w}$ is the remaining part. Define $\mathbf v:=\mathbf {\lambda}-\mathbf {\mu}\;$. It can be shown that optimal values of \eqref{eq::ERO2} satisfy $\mathbf {\lambda}+\mathbf {\mu}=|\mathbf {\lambda}-\mathbf {\mu}|$. Here, $|\cdot|$ denotes the component-wise absolute value.  Therefore, \eqref{eq::ERO2} can be rewritten as an optimization over $\mathbf v\in \R^{m}$ and $\mathbf w\in \R^{2n-1}$ in the following form:
\begin{equation}
\label{eq::ERO3}
\begin{aligned}
\min_{\substack{w_1\in\R\\ \tilde{\w}\in\R^{2n-2}}}\max_{\mathbf v \in \R^m}& 
-{\mathbf{\eta}}^{T}\w+{\mathbf v}^{T}\tilde{\mathbf{A}}\tilde{\w}+{\mathbf v}^{T}\mathbf{A}_1(1+w_1)-{|\mathbf v|}^{T}\mathbf{A}_1.
\end{aligned}
\end{equation}
Note that $\tilde{\mathbf{A}}$ has i.i.d. standard normal entries. One can check that \eqref{eq::ERO3} satisfies the condition of Lemma \ref{lem:CGMT}. Hence, we can form the (AO) as follows,
\begin{equation}
\label{eq::AOERO}
\begin{aligned}
\min_{\substack{w_1\in\R\\ \tilde{\w}\in\R^{2n-2}}}\max_{\mathbf v \in \R^m}& 
-{\mathbf{\eta}}^{T}\w+{\mathbf v}^{T}\mathbf g||\tilde{\w}||+||\mathbf v||\mathbf h^{T}\tilde{\w}\\
&+{\mathbf v}^{T}\mathbf{A}_1(1+w_1)-{|\mathbf v|}^{T}\mathbf{A}_1,
\end{aligned}
\end{equation}
where $\mathbf g\in \R^m$ and $\mathbf h \in \R^{2n-2}$ with entries drawn independently from standard normal distribution. Analysis of \eqref{eq::AOERO} is similar to~\cite{dhifallah2017phase}. Due to lack of space, we defer technical details to the full version of the paper. 

We conclude the paper with a theorem that characterizes the performance of the ERO. Let $\w^*$ be the optimizer of \eqref{eq::AOERO}. Define $s^{*}:=1+w^*_1$ and $t^{*}:=||\tilde{\w}^*||$. 

\begin{thm}
\label{thm::perERO}

In the asymptotic regime where $m,n\rightarrow \infty$, and  $\delta:=\frac{m}{n}$, $s^*$ and $t^*$ converges to the solution of the following deterministic optimization,
\begin{equation}
\label{eq::performanceAO}
\begin{aligned}
\max_{s\in[-1,1],\;\; t\geq 0}& \rho_{\text{init}}\;s+\sqrt{1-{\rho_{\text{init}}}^2}\sqrt{t^2-\frac{\delta}{2}p(t,s)}\\
\text{subject to: }&p(t,s)\leq \frac{2t^2}{\delta}.
\end{aligned}
\end{equation} 
In the above optimization, $p(t,s)$ is define as,
\begin{align}
p(t,s)=&t^2+(1+s)[1+s-\sqrt{t^2+(1+s)^2}]\nonumber \\
&+(1-s)[1-s-\sqrt{t^2+(1-s)^2}]
\end{align}
\end{thm}

It can be shown that $\rho_{\text{init}}>\frac{2}{\sqrt{\delta}}$ is the necessary and sufficient condition for $(t^*,s^*)=(0,1)$ to be the unique solution of \eqref{eq::performanceAO} which is equivalent to the perfect recovery in the ERO.

\section{Acknowledgment}
\label{sec:ack}

This work was inspired by the ideas presented in \cite{dhifallah2017phase}. The authors would like to thank Yue M. Lu, Christos Thrampoulidis, and Philipp Walk for helpful discussions.
\vspace{+5pt}

\bibliographystyle{IEEEbib}
\bibliography{compbib}

\end{document}